\documentclass[conference]{IEEEtran}
\IEEEoverridecommandlockouts

\usepackage{cite}
\usepackage{amsmath,amssymb,amsfonts}
\usepackage{xcolor}
\usepackage{booktabs}
\usepackage{graphicx}
\def\BibTeX{{\rm B\kern-.05em{\sc i\kern-.025em b}\kern-.08em
    T\kern-.1667em\lower.7ex\hbox{E}\kern-.125emX}}

\newtheorem{theorem}{Theorem}
\newtheorem{definition}{Definition}

\newtheorem{example}{Example}

\newtheorem{corollary}{Corollary}
\newenvironment{proof}{\par \textit{Proof:}}

\usepackage{color}
\usepackage{subfig}
\usepackage{arydshln}
\usepackage{algorithm}
\usepackage{algpseudocode}

\usepackage{etoolbox}
\AtBeginEnvironment{algorithm}{\setlength{\intextsep}{3pt}}
\makeatletter
\g@addto@macro\normalsize{%
  \setlength{\abovedisplayskip}{6pt}%
  \setlength{\belowdisplayskip}{6pt}%
  \setlength{\abovedisplayshortskip}{4pt}%
  \setlength{\belowdisplayshortskip}{4pt}%
  \setlength{\arraycolsep}{3pt}   
  \renewcommand{\arraystretch}{0.9} 
}
\makeatother

\usepackage[top=0.75in,bottom=1in,left=0.625in,right=0.625in]{geometry} 

\begin{document}







\title{Universal Costas Matrices: Towards a General Framework for Costas Array Construction\\
\thanks{Fatih Gulec was supported by EU Horizon Europe under the Marie Skłodowska-Curie COFUND grant No 101081327 YUFE4Postdocs.}
}

\author{\IEEEauthorblockN{Fatih Gulec}
\IEEEauthorblockA{\textit{School of Computer Science and Electronic Engineering} \\
\textit{University of Essex, UK, f.gulec@essex.ac.uk}}\\
\and
\IEEEauthorblockN{Vahid Abolghasemi}
\IEEEauthorblockA{\textit{School of Computer Science and Electronic Engineering} \\
\textit{University of Essex, UK, v.abolghasemi@essex.ac.uk}}\\
}

\IEEEaftertitletext{\vspace{-1.25cm}} 
\maketitle

\begin{abstract}
Costas arrays are a special type of permutation matrices with ideal autocorrelation and low cross-correlation properties, making them valuable for radar, wireless communication, and integrated sensing and communication applications. This paper presents a novel unified framework for analyzing and discovering new Costas arrays. We introduce Universal Costas Matrices (UCMs) and Universal Costas Frequency Matrices (UCFMs) and investigate their structural characteristics. A framework integrating UCMs and UCFMs is proposed to pave the way for future artificial intelligence-assisted Costas array discovery. Leveraging the structural properties of UCMs and UCFMs, a reconstruction-based search method is developed to generate UCMs from UCFMs. Numerical results demonstrate that the proposed approach significantly accelerates the search process and enhances structural insight into Costas array generation.
\end{abstract}


\section{Introduction}
Costas arrays are permutation matrices in which the vectors connecting every pair of 1s are distinct. They were originally proposed to be used as frequency-hopping codes in radar/sonar systems \cite{costas1984study}. They exhibit ideal autocorrelation properties, which are desirable for low probability of intercept radars to mitigate Doppler shift effects \cite{gulecc2015usage, yao2024optimum}. These arrays are widely applied in modern radar waveform design research \cite{touati2017multi, correll2019costas,  xing2025phase}. 

On the other hand, Costas arrays have been employed in multiple access spread spectrum wireless communication systems such as orthogonal frequency division multiplexing systems \cite{maric1992class, maric2006using}. Additionally, recent studies show that modified versions of Costas arrays can have low cross-correlation and optimal frequency hopping properties, as in orthogonal pattern division multiple access systems \cite{yao2024orthogonal, yao2019algebraic, yao2024optimum}. These features make Costas arrays attractive candidates for integrated sensing and communication (ISAC) applications \cite{xing2025phase}, with their use in multi-user systems offering a potential solution to the sensing performance improvement challenge discussed in \cite{lu2024integrated}.


While the application areas of Costas arrays are broad, their generation techniques are limited. Conventionally, they can be generated with algebraic methods using Galois fields \cite{golomb1984constructions, drakakis2006review, golomb2007status} and their structural properties are extensively studied \cite{beard2020singular, jedwab2017costas, warnke2022density, torres2023multidimensional}. Researchers also employed exhaustive search techniques for the enumeration of Costas arrays \cite{drakakis2008results, drakakis2011enumeration, drakakis2011results, beard2007costas, russo2010costas}. To the best of our knowledge, all Costas arrays have been enumerated only up to order $29$ \cite{drakakis2011results}. All known Costas arrays, whether obtained algebraically or through exhaustive search, are collected in a publicly available database \cite{beard2017}. Such a database could potentially be employed to enable artificial intelligence (AI) based Costas array generation similar to those used in matrix completion problems \cite{yoon2018gain, wang2023generative}.  



Despite these efforts, no general construction method for Costas arrays of any arbitrary order has been discovered to date. On the one hand, algebraic construction methods remain largely restricted to prime or prime-power orders. On the other hand, exhaustive search rapidly becomes infeasible due to its factorial growth in complexity. These ongoing challenges have motivated us to develop a unified analytical framework that enables the application of AI-assisted search and generative modeling approaches toward generalized Costas array generation. Moreover, such a unified representation could leverage existing datasets of known arrays to facilitate the prospective AI-based generation of new Costas arrays.

In this paper, a unified approach to analyze and search for new Costas arrays is proposed. We define Universal Costas Matrices (UCMs), where each row contains a distinct Costas array of order $n$ in canonical form, and Universal Costas Frequency Matrices (UCFMs), which are $n \times n$ matrices that compactly represent the frequency of each element in the corresponding UCM. The contributions of the paper can be summarized as follows:
\begin{itemize}
    \item UCMs and UCFMs are introduced as novel matrix representations for the analysis and search of Costas arrays.
    \item A theoretical analysis of the structural properties of UCMs and UCFMs is presented.
    \item An AI-enabled framework is proposed, where UCMs and UCFMs can be integrated into future generative AI-based search methods for Costas array discovery.
    \item A novel search method, i.e., the reconstruction of UCMs from UCFMs, which exploits the structural properties of UCMs and UCFMs, is proposed. 
\end{itemize}

Furthermore, UCFM heatmaps are visualized to illustrate their structural symmetry, and the proposed reconstruction method is quantitatively evaluated. The results demonstrate that the proposed approach yields significant runtime improvement when compared with a benchmark method, while providing a foundation for future AI-based Costas array generation.

\section{Universal Costas Matrices} \label{UCM}
In this section, universal Costas matrices (UCMs) and their structural properties are presented.
\subsection{Preliminaries}
\begin{definition}
    A Costas array of order $n$ is defined as an $n \times n$ permutation matrix that satisfies the condition of having distinct vectors connecting every two 1's. The row indices of this permutation matrix give the Costas array.
\end{definition}


\begin{definition} \label{U_n}
Let $C(n)$ be the number of all possible distinct Costas arrays of order $n$. A UCM of order $n$, denoted by $U_n$, is a $C(n) \times n$ matrix comprising all distinct Costas arrays of order $n$ in each row. The rows of $U_n$ are given canonically such that the first element of each row takes integer values from $1$ to $n$, and the remaining elements in each row are arranged accordingly to preserve canonical ordering.
The matrix is divided row-wise into $n$ consecutive blocks, where the $j^{th}$ block contains all Costas arrays whose first element is $j$. Thus, $U_n$ can be shown as
\begingroup
\setlength{\arraycolsep}{3pt} 
\renewcommand{\arraystretch}{0.5} 
\begin{equation}
U_n =
\left[
\begin{array}{cccc}
c_{1,1} & c_{1,2} & \cdots & c_{1,n} \\
\vdots & \vdots & \ddots & \vdots \\
c_{r_1,1} & c_{r_1,2} & \cdots & c_{r_1,n} \\ 
\hdashline[2pt/2pt]
c_{r_1+1,1} & c_{r_1+1,2} & \cdots & c_{r_1+1,n} \\
\vdots & \vdots & \ddots & \vdots \\
c_{r_2,1} & c_{r_2,2} & \cdots & c_{r_2,n} \\ 
\hdashline[2pt/2pt]
\vdots \\[-2pt]
\hdashline[2pt/2pt]
c_{r_{n-1}+1,1} & c_{r_{n-1}+1,2} & \cdots & c_{r_{n-1}+1,n} \\ 
\vdots & \vdots & \ddots & \vdots \\
c_{C(n),1} & c_{C(n),2} & \cdots & c_{C(n),n}
\end{array}
\right],
\end{equation}
\endgroup
where dashed lines separate canonical blocks, and $r_j$ denotes the last row index of the $j^{th}$ block corresponding to Costas arrays whose first element equals $j$. Formally, a UCM can be denoted as $U_n = \mathcal{U}(n)$, where $\mathcal{U}(\cdot)$ is the \textit{UCM construction operator} that maps the set of all distinct Costas arrays of order $n$ to a canonical matrix representation,
\begin{equation}
   \mathcal{U}(n) : \{\text{All Costas arrays of order } n\} \rightarrow \mathbb{N}^{C(n) \times n}, 
\end{equation}
such that the rows of $\mathcal{U}(n)$ are ordered by their first element in ascending order, i.e.,$ \quad U_n(r_1,1) < U_n(r_2,1), \forall\, r_1 < r_2 $ to ensure a canonical block structure.
\end{definition}

\begin{example} \label{ex_U_n}
    $U_4$ can be shown in canonical order with $4$ blocks separated by horizontal dashed lines as follows:
    \begin{equation}
        U_4 =
        \left[
        \begin{array}{cccc}
            1 & 2 & 4 & 3 \\
            1 & 3 & 4 & 2 \\
            1 & 4 & 2 & 3 \\[3pt]
            \hdashline[2pt/2pt]
            2 & 1 & 3 & 4 \\
            2 & 3 & 1 & 4 \\
            2 & 4 & 3 & 1 \\[3pt]
            \hdashline[2pt/2pt]
            3 & 1 & 2 & 4 \\
            3 & 2 & 4 & 1 \\
            3 & 4 & 2 & 1 \\[3pt]
            \hdashline[2pt/2pt]
            4 & 1 & 3 & 2 \\
            4 & 2 & 1 & 3 \\
            4 & 3 & 1 & 2 \\
        \end{array}
        \right].
    \end{equation}

\end{example}
\begin{definition}
    Based on the Definition \ref{U_n}, $S(n,k)$ is defined as the sum of the $k^{th}$ column of $U_n$ with elements $c_{m,k}$ where $k \in {1, 2, ..., n}$ as given by \vspace{-0.3cm}
    \begin{equation} \vspace{-0.2cm}
        S(n, k) = \sum_{m=1}^{C(n)} c_{m,k}.
    \end{equation}
\end{definition}

In addition, it is useful to recall the elementary identity that for any permutation $\pi$ of $\{1,\dots,n\}$ we have $\sum_{k=1}^n \pi(k)=\frac{n(n+1)}{2}$. Since a Costas array is a permutation, the sum of any Costas array of order $n$ or the sum of the elements in any row in $U_n$ ($D(n)$) is given by \vspace{-0.3cm}
\begin{equation} \label{D_n} 
    D(n) = \frac{n (n+1)}{2}.
\vspace{-0.3cm}
\end{equation}

\begin{table}[t!]
\centering
\caption{Number of Costas Arrays and Row/Column Sums in Universal Costas Matrices}
\label{Table_C_S}
\renewcommand{\arraystretch}{1.15}
\setlength{\tabcolsep}{6pt}
\begin{tabular}{cccc|cccc}
\hline
\multicolumn{4}{c|}{\textbf{Orders 3--16}} & \multicolumn{4}{c}{\textbf{Orders 17--29}} \\ 
\hline
$\mathbf{n}$ & $\mathbf{C(n)}$ & $\mathbf{D(n)}$ & $\mathbf{S(n)}$ & 
$\mathbf{n}$ & $\mathbf{C(n)}$ & $\mathbf{D(n)}$ & $\mathbf{S(n)}$ \\ 
\hline
3  & 4    & 6   & 8     & 17 & 18276 & 153 & 164484 \\
4  & 12   & 10  & 30    & 18 & 15096 & 171 & 143412 \\
5  & 40   & 15  & 120   & 19 & 10240 & 190 & 102400 \\
6  & 116  & 21  & 406   & 20 & 6464  & 210 & 67872  \\
7  & 200  & 28  & 800   & 21 & 3536  & 231 & 38896  \\
8  & 444  & 36  & 1998  & 22 & 2052  & 253 & 23598  \\
9  & 760  & 45  & 3800  & 23 & 872   & 276 & 10464  \\
10 & 2160 & 55  & 11880 & 24 & 200   & 300 & 2500   \\
11 & 4368 & 66  & 26208 & 25 & 88    & 325 & 1144   \\
12 & 7852 & 78  & 51038 & 26 & 56    & 351 & 756    \\
13 & 12828 & 91 & 89796 & 27 & 204   & 378 & 2856   \\
14 & 17252 & 105 & 129390 & 28 & 712 & 406 & 10324 \\
15 & 19612 & 120 & 156896 & 29 & 164 & 435 & 2460  \\
16 & 21104 & 136 & 179384 &  --   &   --   &   --  &  --    \\
\hline
\end{tabular}
\end{table}

\subsection{Structural Properties and Theoretical Results}

\begin{theorem} \label{Theorem_UCM}
In a universal Costas matrix $U_n$, all columns have equal sums, i.e.,
\begin{equation}
S(n,k)=S(n), \qquad \forall k\in\{1,\dots,n\}.
\end{equation}
\end{theorem}

\begin{proof}
Each Costas array comes in sets of four or eight since they belong to an equivalence class of polymorphs generated by the eight reflections and rotations \cite{drakakis2006review}. Depending on diagonal symmetry, each polymorph set, i.e., orbit, has four or eight members obtained by $90^\circ$ rotations or reflections of a Costas array permutation $\pi$ and its inverse $\pi^{-1}$. For a fixed column $k$, four rotated polymorphs can be expressed as
\begin{equation} \label{eq:rotations}
\setlength{\jot}{1pt} 
\begin{aligned}
\pi^{(1)}(k)&=\pi(k),\!\!  &\pi^{(2)}(k)&=n{+}1{-}\pi(k),\\[-1pt]
\pi^{(3)}(k)&=\pi(n{+}1{-}k),\!\! &\pi^{(4)}(k)&=n{+}1{-}\pi(n{+}1{-}k),
\end{aligned}
\end{equation}
where the other four reflection polymorphs are obtained by replacing $\pi(k)$ with $\pi^{-1}(k)$. Each pair within the orbit is complementary:
\begin{equation}
\pi^{(i)}(k)+\pi^{(i+1)}(k)=n+1, \qquad i\in\{1,3,5,7\}.
\end{equation}
Hence, the column sum over one set of polymorphs is
\begin{equation}
\sum_{t\in O(\pi)} \pi^{(t)}(k)=\frac{|O(\pi)|}{2}(n+1),
\end{equation}
where $|O(\pi)|\!\in\!\{4,8\}$ is the orbit size depending on the diagonal symmetry. This value depends only on $n$, not on $k$. Thus, stacking all polymorphs to form $U_n$ yields
\begin{equation}
S(n,k)=S(n), \quad \forall k.
\end{equation}
\end{proof}

\begin{corollary} \label{Corollary_Sn_Cn_short}
The ratio of the sum of a column in a UCM of order $n$ to the number of Costas arrays of order $n$ is given by
\begin{equation}
    \frac{S(n)}{C(n)}=\frac{n+1}{2}.
\end{equation}
\end{corollary}

\begin{proof}
Summing all entries of $U_n$ by rows gives $C(n)D(n)$, while summing by columns gives $nS(n)$. Equating and using $D(n)=\tfrac{n(n+1)}{2}$ in (\ref{D_n}) gives the result:
\begin{align} 
\setlength{\jot}{1pt} 
C(n)\frac{n(n+1)}{2} &= nS(n)  \\
\frac{S(n)}{C(n)} &= \frac{n+1}{2}. \setlength{\jot}{1pt} 
\end{align}
\setlength{\jot}{1pt} 
\end{proof}

\section{Universal Costas Frequency Matrices} \label{UCFM}
In this section, we present the universal Costas frequency matrices (UCFMs) and their symmetrical properties.
\begin{definition}
    Let $U_n$ be the UCM of order $n$. A UCFM of order $n$, denoted by $F_n$, is an $n \times n$ matrix that shows how frequently each number from $1$ to $n$ appears in each column of $U_n$. The $(i,k)^{th}$ element of $F_n$, i.e., $f_{i,k}$, counts the number of times the integer $i$ occurs in the $k^{th}$ column of $U_n$. 
    Formally, the UCFM is defined as $F_n = \mathcal{F}(U_n)$ where $\mathcal{F}(\cdot)$ is the \textit{UCFM transformation operator} that maps $U_n$ to its corresponding frequency matrix $F_n$, defined elementwise as
\begin{equation} \label{eq:F_def}
[\mathcal{F}(U_n)]_{i,k} = f_{i,k} = \#\{\,r \mid U_n(r,k) = i\,\},
\end{equation}
where $\#\{\,\cdot\,\}$ denotes the cardinality of a set, i.e., the number of elements satisfying a given condition.  For a better understanding, two samples of even and odd-numbered order UCFMs are given below.
\end{definition}

\begin{example} \label{ex_F_n}
    $F_6$ is given by
    \begin{equation}
        F_6 = 
        \begin{bmatrix}
            19 & 17 & 22 & 22 & 17 & 19 \\
            17 & 24 & 17 & 17 & 24 & 17 \\
            22 & 17 & 19 & 19 & 17 & 22 \\
            22 & 17 & 19 & 19 & 17 & 22 \\
            17 & 24 & 17 & 17 & 24 & 17 \\
            19 & 17 & 22 & 22 & 17 & 19
        \end{bmatrix},
    \end{equation}
    where the first row shows the number of $1$'s in each column of $U_6$, the second row shows the number of $2$'s in each column of $U_6$, and so on. Similarly, $F_5$ is also shown below.
    \begin{equation}
        F_5 = 
        \begin{bmatrix}
            6 & 10 & 8 & 10 & 6 \\
            10 & 6 & 8 & 6 & 10 \\
            8 & 8 & 8 & 8 & 8 \\
            10 & 6 & 8 & 6 & 10 \\
            6 & 10 & 8 & 10 & 6
        \end{bmatrix},
    \end{equation}
\end{example}
where for instance, $f_{2,2} = 6$ shows the number of $2$'s in column $2$ of $U_5$. $F_6$ and $F_5$ shown in Example \ref{ex_F_n} also illustrate the two different symmetries as generalized below.

\begin{theorem} \label{thm_UCFM_sym}
Let $F_n$ be the UCFM of order $n$. $F_n$ shows horizontal, vertical, and diagonal reflection symmetries. For odd $n$, the symmetry axes pass through the middle row and column, while for even $n$, they lie between them.
\end{theorem}

\begin{proof}
As shown in Theorem \ref{Theorem_UCM}, each Costas array (each row of $U_n$) belongs to an equivalence class of polymorphs (rotations and reflections) of size $4$ or $8$ depending on diagonal symmetry \cite{drakakis2011results}. Let each array $\pi^{(r)}$ be represented by its 2-D binary form $B^{(r)}=[b^{(r)}_{i,k}]$ with $b^{(r)}_{i,k}=1$ if $\pi^{(r)}(i)=k$, and $0$ otherwise. $F_n$ can be defined as the superposition of all such matrices as given by
\begin{equation} \label{eq:F_sup}
F_n = \sum_r B^{(r)}, \qquad f_{i,k} = \sum_r b^{(r)}_{i,k}.
\end{equation}

For every $B^{(r)}$, its rotated and reflected polymorphs relocate the entries of $1$ to positions related by index transformations of the form
$(i,k)\mapsto(n{+}1{-}i,k),\ (i,n{+}1{-}k),\ (n{+}1{-}i,n{+}1{-}k),\ (k,i),$ and their complementary reflections across the two diagonals. 
As a result, every value $(i,k)$ in $F_n$ has equal contributions with its symmetric counterparts:
\begin{equation}
\setlength{\jot}{1pt} 
\begin{aligned}
f_{i,k}
&= \! f_{\,n+1-i,\,k}
= \! f_{\,i,\,n+1-k}
= \! f_{\,n+1-i,\,n+1-k}
= \! f_{k,\,i} \\[-2pt]
&= \! f_{\,n+1-k,\,i}
= \! f_{k,\,n+1-i}
= \! f_{\,n+1-k,\,n+1-i},
\quad \forall i,k.
\end{aligned}
\end{equation}
If a Costas array is symmetric along a diagonal (four polymorphs), some of these positions coincide, reducing the number of distinct equalities. Thus, the total contribution of all polymorphs remains balanced, and the element counts at symmetric positions in $F_n$ become equal.
\end{proof}

\begin{corollary} \label{cor:Cn_F}
    The sum of all elements in a row or column of $F_n$ gives the number of all Costas arrays of order $n$, i.e.,
    \begin{equation} \label{Cn_F}
        C(n) = 
        \begin{cases}
            \sum_{k=1}^{n} f_{k,m}, \forall m \\
            \sum_{m=1}^{n} f_{k,m}, \forall k.
        \end{cases}
    \end{equation}
\end{corollary}

\begin{proof}
    Given that $F_n$ counts the frequency of each number in each column and it has $C(n)$ rows, the sum of each column in $F_n$ is $C(n)$, since each number $1$ to $n$ appears only once in a Costas array of order $n$. Additionally, the reflection symmetry of $F_n$ given in Theorem \ref{thm_UCFM_sym} provides that the sum of each column equals the sum of each row in $F_n$, which is also equal to $C(n)$. 
\end{proof}

\begin{corollary} 
    By using the elements of $F_n$, the sum of each column of $U_n$, i.e., $S(n)$, can be given by
    \begin{equation} \label{eq_S_F}
        S(n) = \sum_{k=1}^{n} k f_{k,m}, \hspace{0.1cm}\forall m
    \end{equation}
    where $f_{k,m}$ is the element of $F_n$ at the row $k$ and column $m$.
\end{corollary}

\begin{proof}
Each $f_{k,m}$ counts how many times $k$ appears in column $m$ of $U_n$. 
Multiplying each count by $k$ and summing over all rows gives the column sum $S(n)$, which is the same for all $m$ (Theorem~\ref{Theorem_UCM}).
\end{proof}

The theoretical foundations given in Sections \ref{UCM} and \ref{UCFM} form the basis of the framework proposed in the next section.



\section{A General Framework for Costas Array Discovery: Reconstruction of UCMs from UCFMs}
\label{U/F}
This section introduces a general framework for Costas array discovery that integrates UCMs, UCFMs, and AI-based generative models. Within this framework, a novel search method based on reconstructing UCMs from UCFMs is proposed and detailed below.

\subsection{Motivation}
As described in Section \ref{UCFM}, UCFMs provide a compact way to represent all Costas arrays of a given order, with 2-D $n \hspace{-0.1cm} \times \hspace{-0.1cm} n$ unique matrices. They exhibit highly structured and symmetrical patterns, that can be exploited to reconstruct the underlying Costas arrays. Furthermore, they can also be represented as images via heatmaps (illustrated in Section \ref{RD}). 



With the advancement of AI techniques, subtle patterns in images can be distinguished. In particular, generative models such as generative adversarial networks \cite{goodfellow2014generative, yoon2018gain} have demonstrated strong capabilities in synthesizing or extrapolating structured data distributions \cite{wang2023generative}. However, AI has not yet been applied to Costas array discovery due to the lack of a well-defined representation that connects Costas arrays with data-driven learning. Considering the compact representation of UCFMs, there is significant potential to leverage AI-based generative modeling to predict UCFMs of higher orders, thereby facilitating the discovery of new Costas arrays. To this end, we propose a general framework that integrates UCMs, UCFMs, and AI methods, as illustrated in Fig. \ref{fig:Block_U_F_AI}. As shown in the upper branch of Fig. \ref{fig:Block_U_F_AI}, all enumerated Costas arrays via exhaustive search can be converted into UCFMs ($F_n$), which are complete since all Costas arrays of order $n \leq 29$ are known. Additionally, algebraic methods can be used to generate incomplete UCFMs ($F^*_n$) using the large Costas array database available up to order $n \leq 1030$ \cite{beard2017}, as shown in the lower branch of Fig. \ref{fig:Block_U_F_AI}. These complete and incomplete UCFMs can thus serve as structured training data for potential generative AI models.

The predicted UCFMs can be obtained as the output of such models, as shown with $\hat{F}_n$ in Fig. \ref{fig:Block_U_F_AI}. Then, these $\hat{F}_n$s can be used to reconstruct the corresponding UCM ($\hat{U}_n$) via the inverse operation $\mathcal{F}^{-1}(\cdot)$, which represents the reconstruction of a UCM from UCFM. This reconstruction perspective provides a more holistic and computationally tractable approach to discovering new Costas arrays, as it leverages frequency constraints rather than direct combinatorial enumeration. We leave the prediction of UCFMs of arbitrary order $n$ using generative AI models as future work, as it is non-trivial and beyond the scope of this paper. However, in this paper, we focus on the final block of Fig. \ref{fig:Block_U_F_AI}, $\mathcal{F}^{-1}(\cdot)$, and present our detailed algorithm for the reconstruction process.
\begin{figure}[b]
    \centering
    \includegraphics[width=0.99\columnwidth]{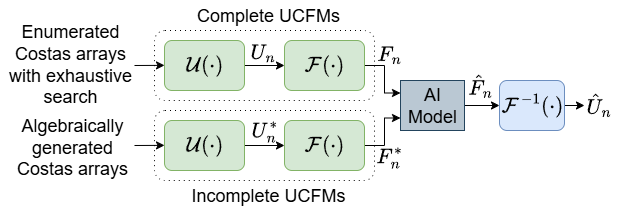}
    \caption{A general framework to discover new Costas arrays using the UCMs/UCFMs and AI.}
    \label{fig:Block_U_F_AI}
\end{figure}

\begin{algorithm}[tb]
\caption{Reconstruction of the Universal Costas Matrix $U_n$ from the Universal Costas Frequency Matrix $F_n$}
\label{alg:reconstructU}
\begin{algorithmic}[1]
\Require $n$: Order; $F_n$: Complete UCFM of size $n \times n$
\Ensure $U_n$: Canonical UCM of order $n$

\State $C(n) \gets \sum_{v=0}^{n-1} F_n[v, 0]$ \Comment{$C(n)$ from (\ref{Cn_F})}
\State Initialize empty matrix $U_n[C(n)][n]$
\State Initialize block sizes $B[v] \gets F_n[v,0]$ for $v = 0 \ldots n-1$
\State Compute block start indices $\text{start}[v]$ and counters $\text{next}[v] \gets \text{start}[v]$, initialize \textit{PlacedRows} $\gets \emptyset$
\For{$a_0 = 0$ \textbf{to} $n-1$} 
    \If{$B[a_0] = 0$} 
        \State \textbf{continue} \Comment{Block already filled}
    \EndIf
    \State $G \gets$ \Call{Russo\_Search}{$n, a_0$} \Comment{Search Costas arrays of order $n$ with first element $a_0$ with Russo's method \cite{russo2010costas}}
    \For{\textbf{each} canonical permutation $p$ \textbf{in} $G$}
        \If{$p \in$ \textit{PlacedRows}} 
            \State \textbf{continue} 
        \EndIf
        \State $\mathcal{O}(p) \gets$ \Call{Polymorphs}{$p$} \hspace{-1cm} \Comment{Generate polymorphs}
        \For{\textbf{each} $q$ \textbf{in} $\mathcal{O}(p)$}
            \If{$q \notin$ \textit{PlacedRows} \textbf{and} $B[q[0]] > 0$}
                \State Place $q$ into block at position $\text{next}[q[0]]$
                \State $\text{next}[q[0]] \gets \text{next}[q[0]] + 1$ 
                \State $B[q[0]] \gets B[q[0]] - 1$ 
                \State \textit{PlacedRows} $\gets$ \textit{PlacedRows} $\cup \{q\}$
            \EndIf
        \EndFor
        \If{$B[a_0] = 0$} 
            \State \textbf{break} \Comment{Block fully filled; move to next block}
        \EndIf
    \EndFor
\EndFor
\State \Return $U_n$
\end{algorithmic}
\end{algorithm}

\subsection{Reconstruction Algorithm}

The pseudocode of the reconstruction algorithm of a UCM from a complete UCFM, represented by the blue box in Fig. \ref{fig:Block_U_F_AI} with the operation $\mathcal{F}^{-1}(\cdot)$, is given in Algorithm \ref{alg:reconstructU}. In this algorithm, we first determine the number of Costas arrays of the given order $n$ from $F_n$ using Corollary \ref{cor:Cn_F}. Additionally, the block sizes and indices of $U_n$, where each block represents a canonical group of Costas arrays starting with a different element (see Definition \ref{U_n} and Example \ref{ex_U_n}), are determined using the first column of $F_n$. Then, a canonical exhaustive search is initiated using the benchmark Costas array search method proposed by Russo et al. \cite{russo2010costas} as implemented by \textit{$RUSSO\_SEARCH$($n, a0$)} in Algorithm \ref{alg:reconstructU}. Russo’s method employs symmetry reduction, bitmasking, forward checking (or look-ahead), and backtracking to accelerate the search. In particular, it encodes the difference triangles of the permutation using bitmasking, while the look-ahead process ensures that at least one valid future column position remains available without violating the uniqueness of the difference triangle. This allows pruning of infeasible branches early in the search.

Each Costas array generates up to $8$ polymorphs, resulting in $3$ or $7$ additional arrays depending on the diagonal symmetry \cite{drakakis2006review}. As applied in Russo's method, we generate polymorphs from each found Costas array. However, we treat them differently. Each generated polymorph is placed into the corresponding block as a row in $U_n$. When the corresponding block in $U_n$ is fully filled, we stop the search for that block and move to the next block. When all the blocks are filled based on the block sizes in $F_n$, the reconstruction is completed. Next, the performance of this reconstruction method is demonstrated, and the theoretical results are validated.

\begin{figure}[tb]
    \centering
    \includegraphics[width=0.99\columnwidth]{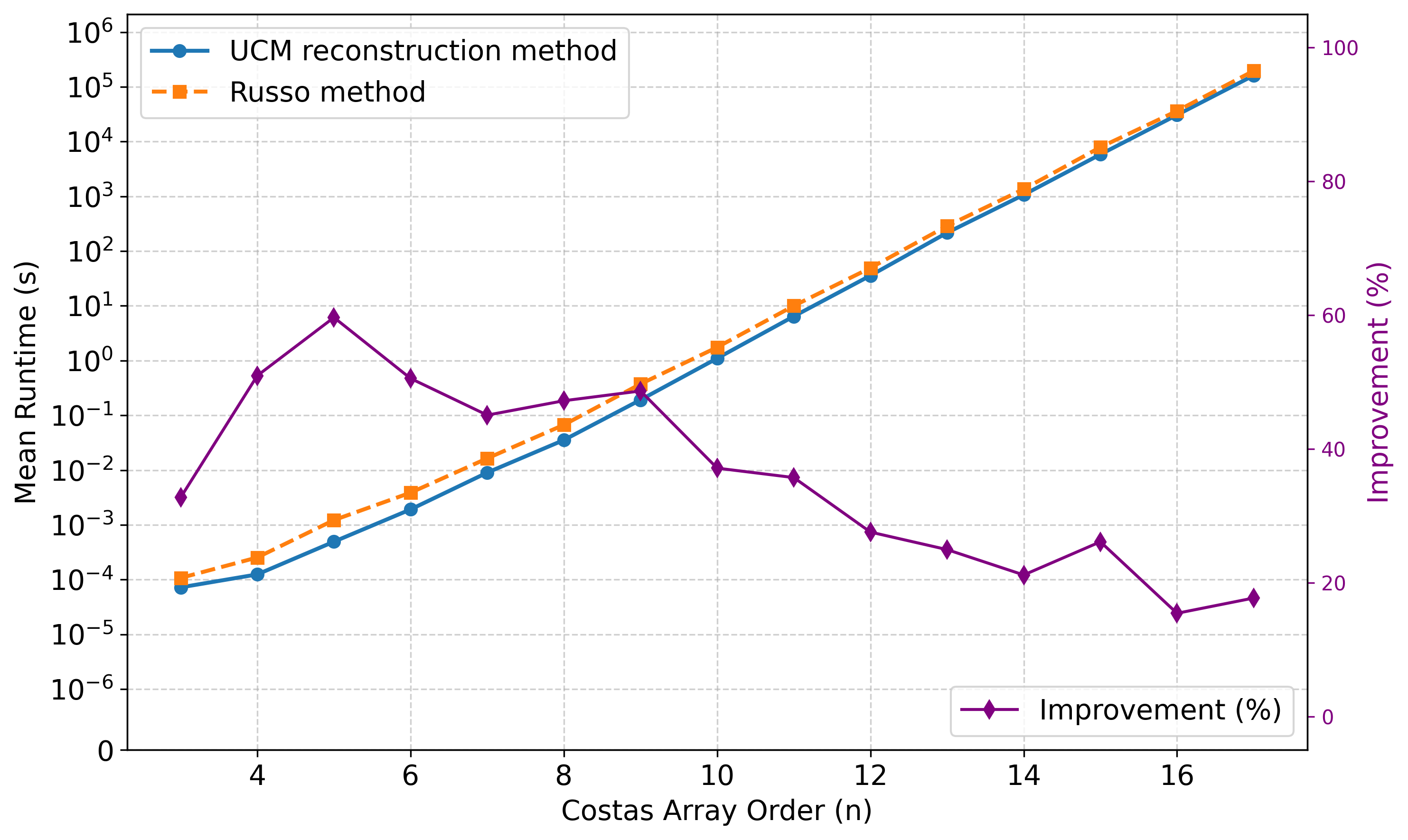}
    \caption{Benchmark comparison results of the proposed UCM reconstruction method with the Russo method.}
    \label{fig:U/F_Results}
\end{figure}

\section{Results} \label{RD}
\vspace{-0.2cm}
In this section, numerical results to illustrate the theoretical findings as well as the performance results of the proposed reconstruction method elaborated in Section \ref{U/F} are presented.

In Fig. \ref{fig:UCFM_all}, the complete UCFMs for $5 \leq n\leq 29$ are illustrated as images via heatmaps in the first 5 rows of the figure. These images are generated by applying (\ref{eq:F_def}) to the enumerated Costas arrays in \cite{beard2017, drakakis2011enumeration, drakakis2011results}. Additionally, incomplete UCFMs are generated in row 6 of Fig. \ref{fig:UCFM_all} using Beard's database of algebraically generated Costas arrays \cite{beard2017}. In Fig. \ref{fig:UCFM_all}, the color scale is normalized based on the maximum UCFM value in each image, where brighter and darker regions indicate higher frequencies and lower frequencies, respectively. Rows $1-5$ of this figure validate the horizontal, vertical, and diagonal reflection symmetries of UCFMs given in Theorem \ref{thm_UCFM_sym}. Incomplete UCFMs still show the same symmetries since the database includes the polymorphs of the algebraically constructed Costas arrays.

\begin{figure*}[t]
    \centering
    \setlength{\tabcolsep}{1pt} 
    \renewcommand{\arraystretch}{0} 
    \begin{tabular}{ccccc}
        \begin{minipage}{0.18\textwidth}\centering \textbf{$n=5$}\\[2pt]
        \includegraphics[width=\linewidth]{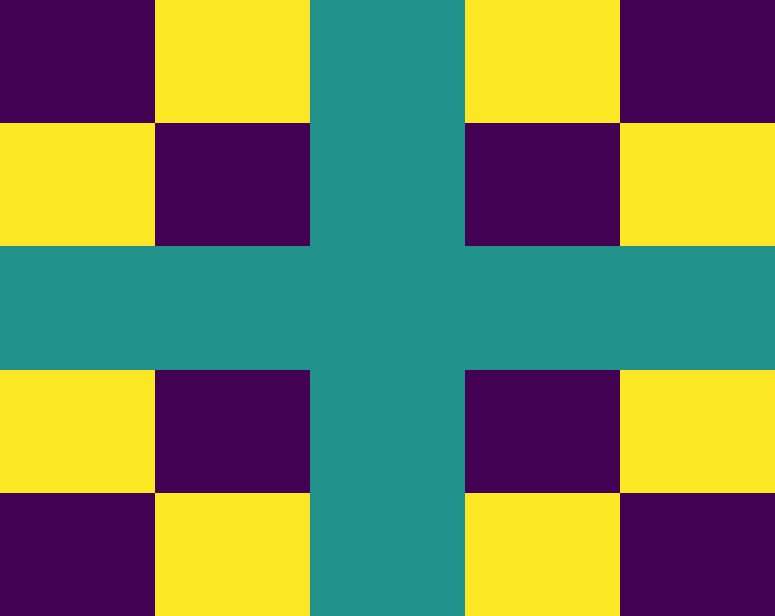}\end{minipage} &
        \begin{minipage}{0.18\textwidth}\centering \textbf{$n=6$}\\[2pt]
        \includegraphics[width=\linewidth]{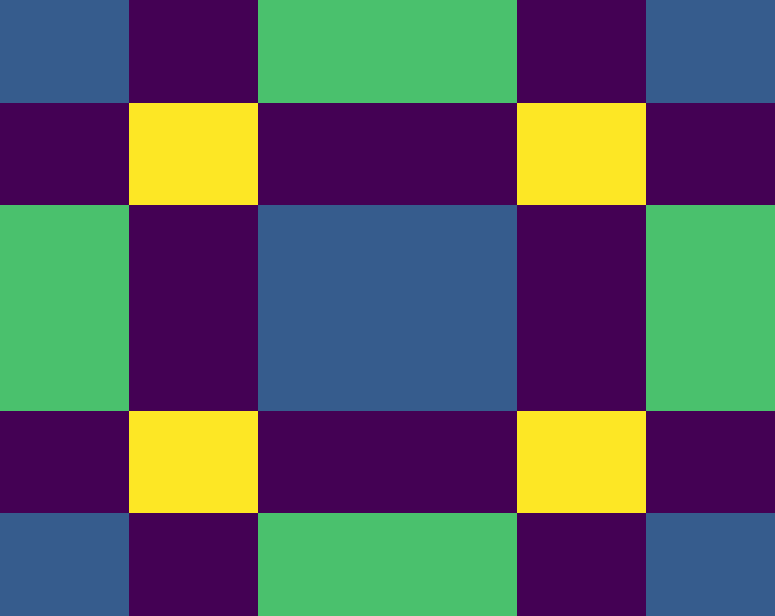}\end{minipage} &
        \begin{minipage}{0.18\textwidth}\centering \textbf{$n=7$}\\[2pt]
        \includegraphics[width=\linewidth]{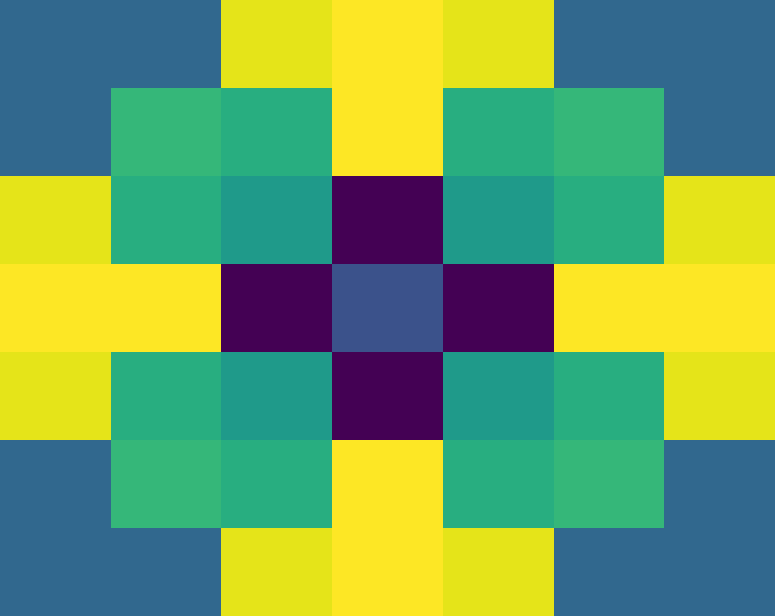}\end{minipage}  &     
        \begin{minipage}{0.18\textwidth}\centering \textbf{$n=8$}\\[2pt]
        \includegraphics[width=\linewidth]{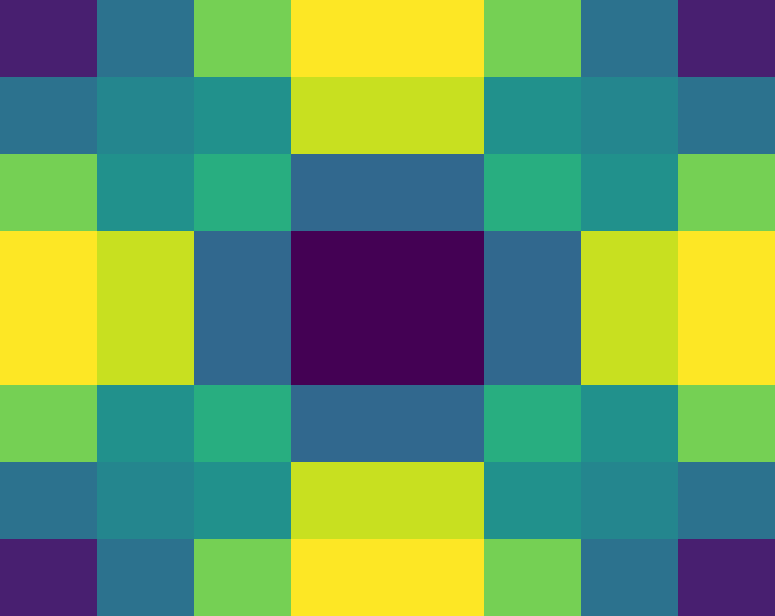}\end{minipage} &
        \begin{minipage}{0.18\textwidth}\centering \textbf{$n=9$}\\[2pt]
        \includegraphics[width=\linewidth]{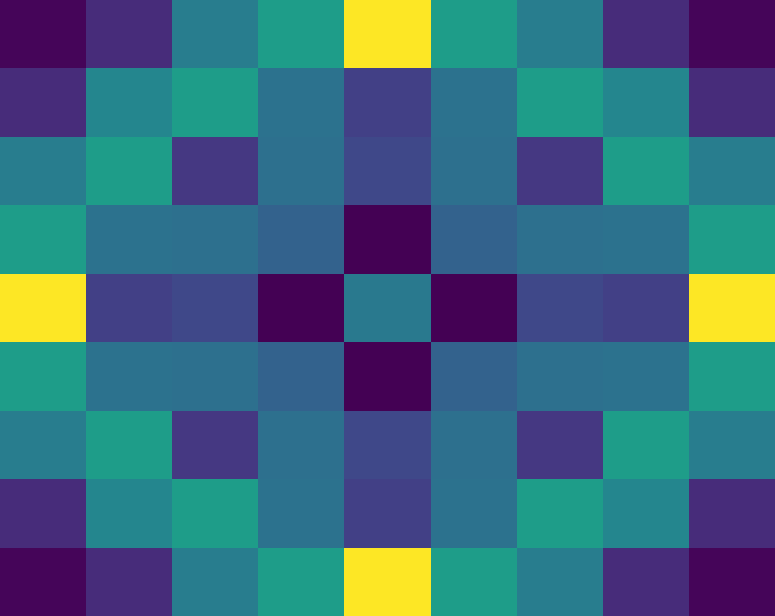}\end{minipage} \\
        \begin{minipage}{0.18\textwidth}\centering \textbf{$n=10$}\\[2pt]
        \includegraphics[width=\linewidth]{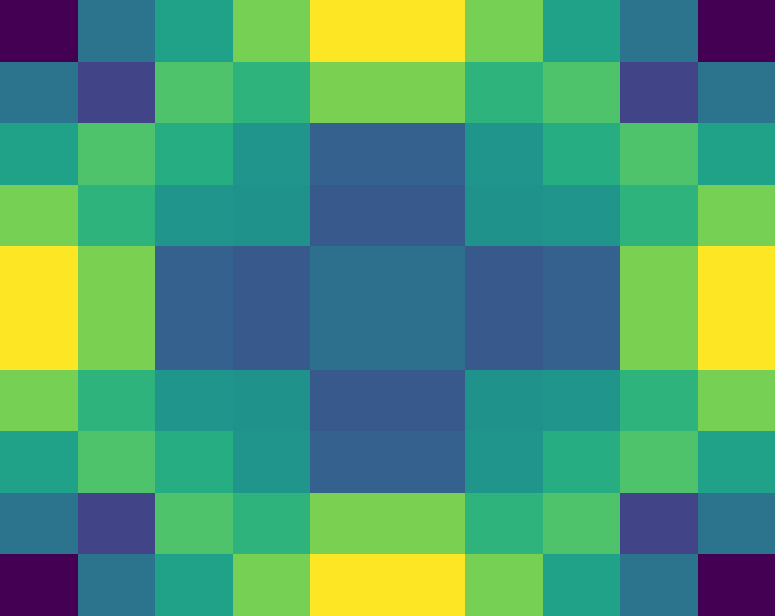}\end{minipage} &
        \begin{minipage}{0.18\textwidth}\centering \textbf{$n=11$}\\[2pt]
        \includegraphics[width=\linewidth]{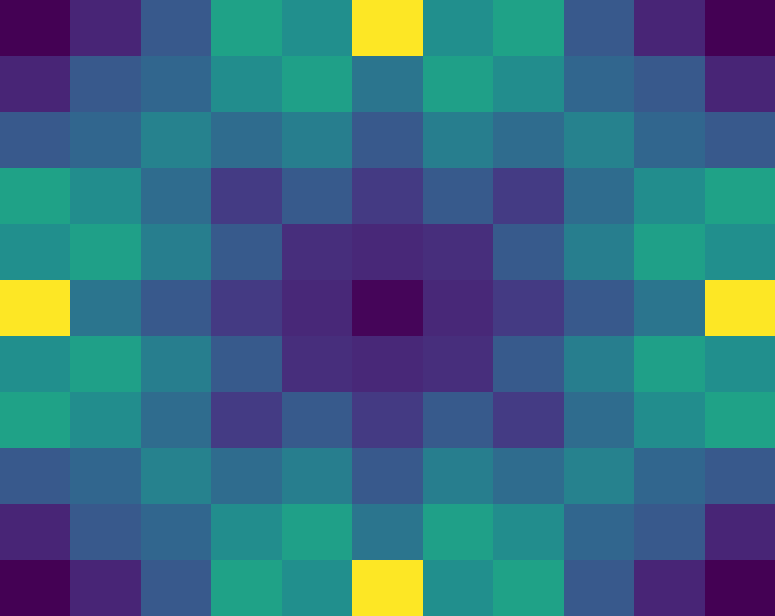}\end{minipage} &
        \begin{minipage}{0.18\textwidth}\centering \textbf{$n=12$}\\[2pt]
        \includegraphics[width=\linewidth]{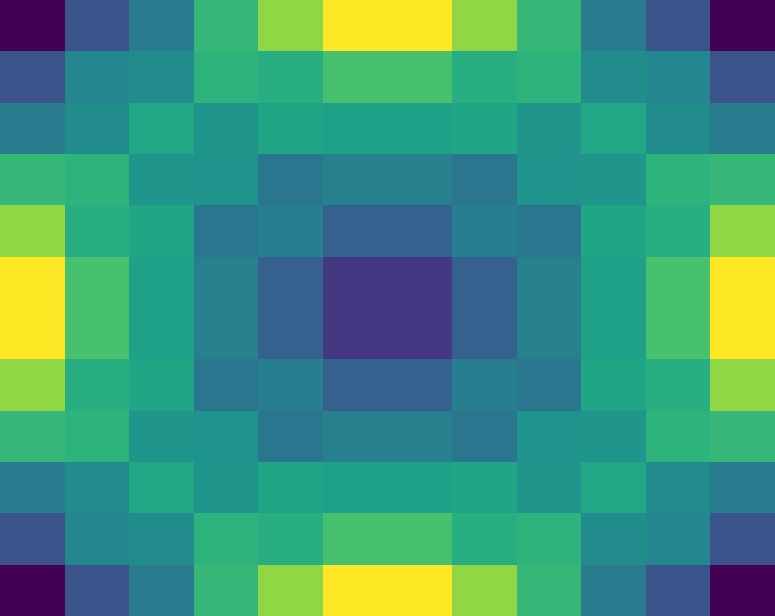}\end{minipage} &
        \begin{minipage}{0.18\textwidth}\centering \textbf{$n=13$}\\[2pt]
        \includegraphics[width=\linewidth]{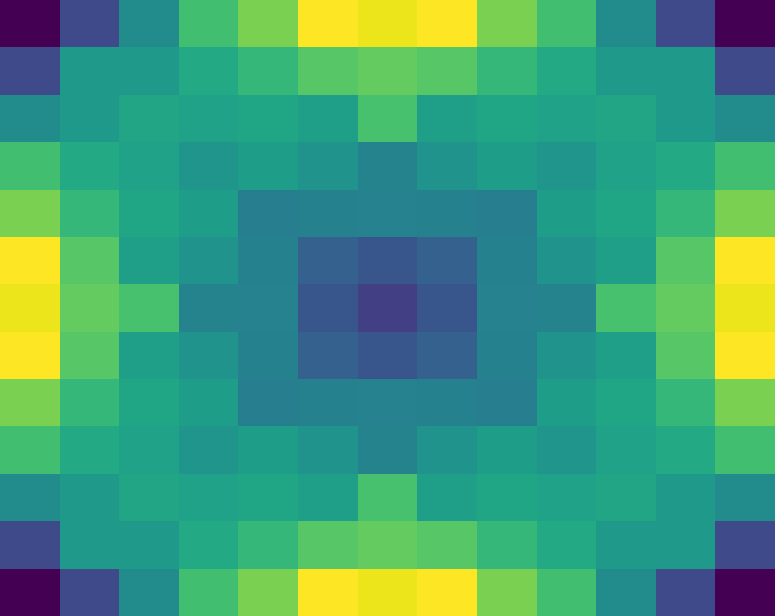}\end{minipage} &
        \begin{minipage}{0.18\textwidth}\centering \textbf{$n=14$}\\[2pt]
        \includegraphics[width=\linewidth]{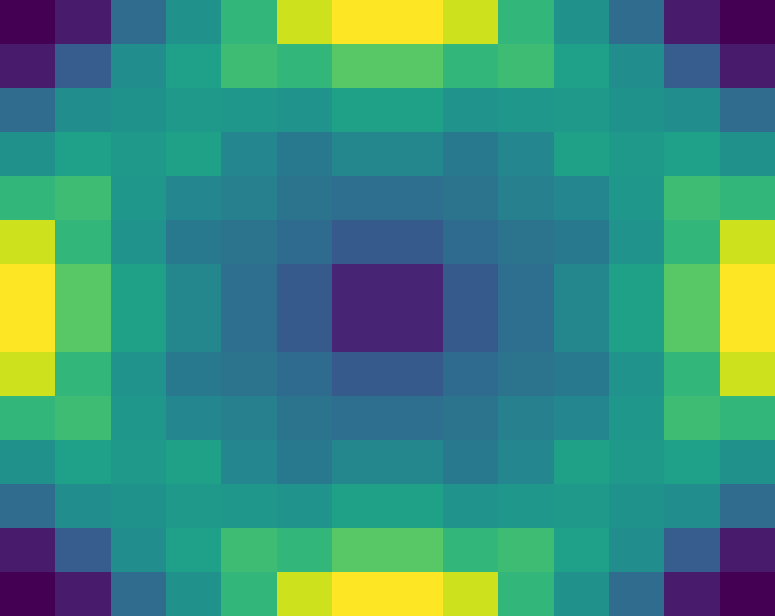}\end{minipage} \\
        \begin{minipage}{0.18\textwidth}\centering \textbf{$n=15$}\\[2pt]
        \includegraphics[width=\linewidth]{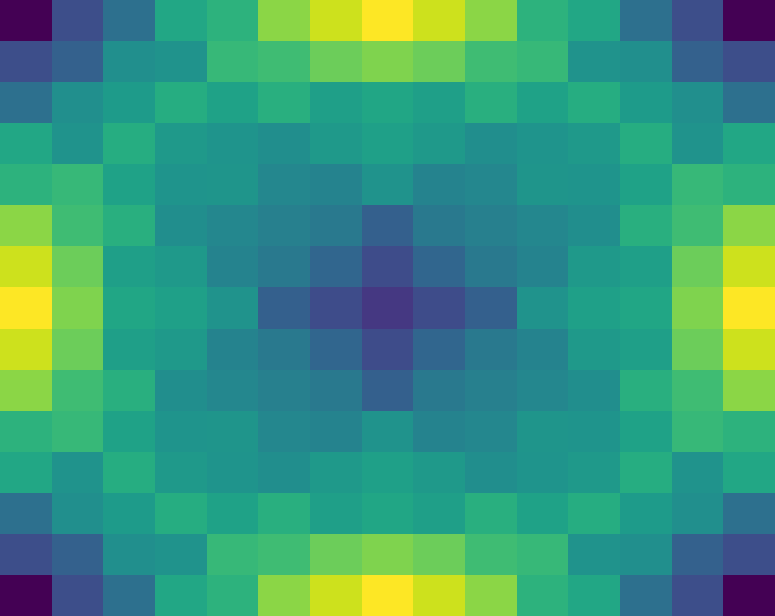}\end{minipage} &
        \begin{minipage}{0.18\textwidth}\centering \textbf{$n=16$}\\[2pt]
        \includegraphics[width=\linewidth]{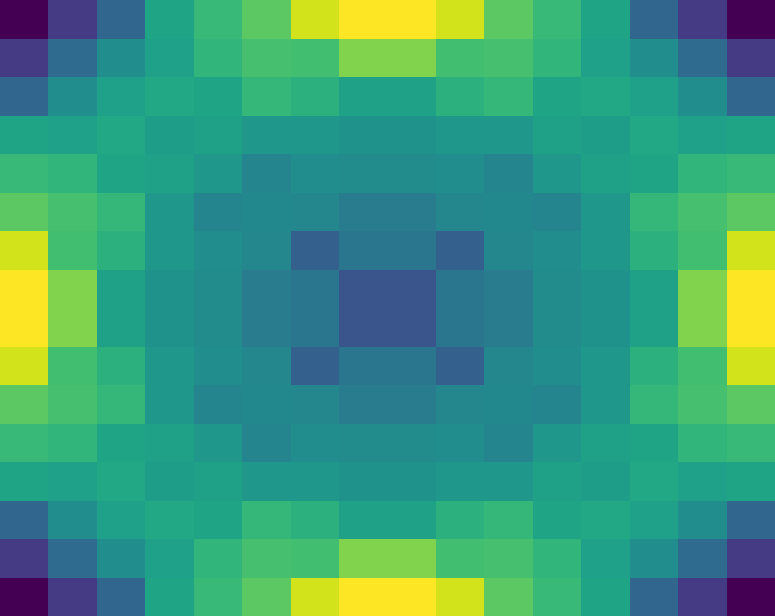}\end{minipage} &
        \begin{minipage}{0.18\textwidth}\centering \textbf{$n=17$}\\[2pt]
        \includegraphics[width=\linewidth]{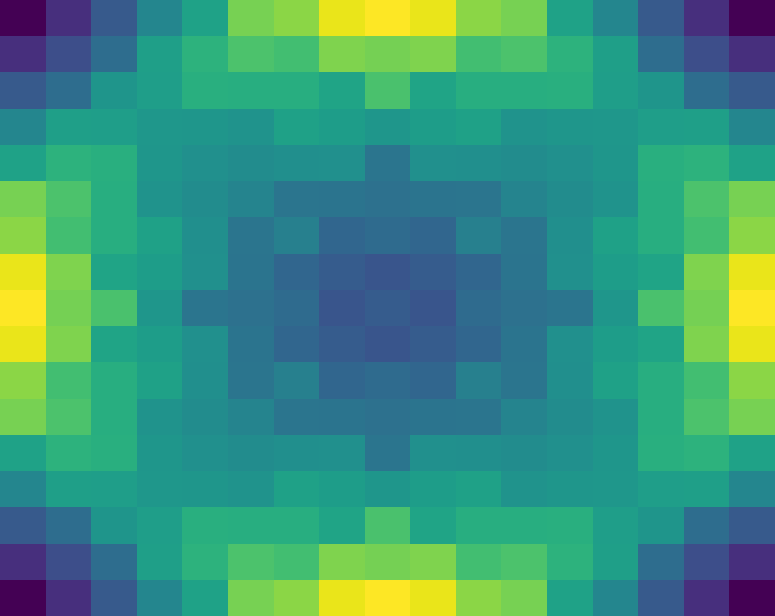}\end{minipage} &
        \begin{minipage}{0.18\textwidth}\centering \textbf{$n=18$}\\[2pt]
        \includegraphics[width=\linewidth]{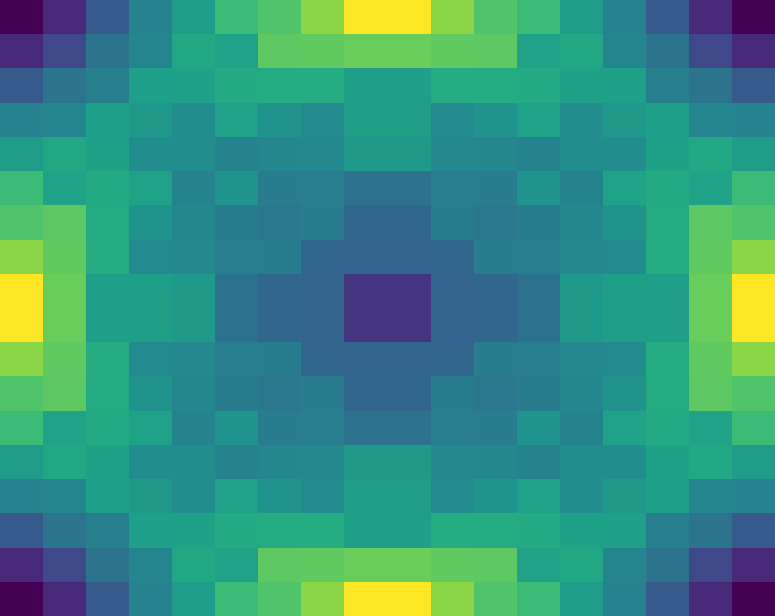}\end{minipage} &
        \begin{minipage}{0.18\textwidth}\centering \textbf{$n=19$}\\[2pt]
        \includegraphics[width=\linewidth]{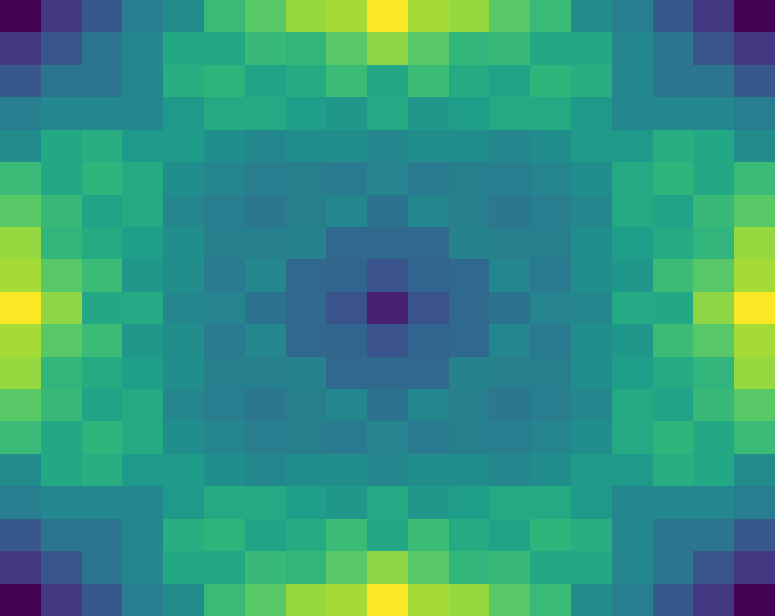}\end{minipage} \\
        \begin{minipage}{0.18\textwidth}\centering \textbf{$n=20$}\\[2pt]
        \includegraphics[width=\linewidth]{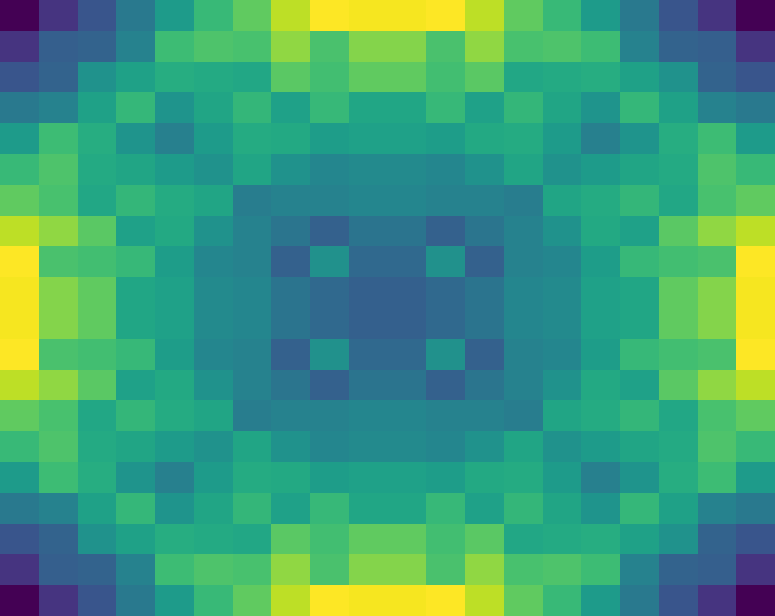}\end{minipage} &
        \begin{minipage}{0.18\textwidth}\centering \textbf{$n=21$}\\[2pt]
        \includegraphics[width=\linewidth]{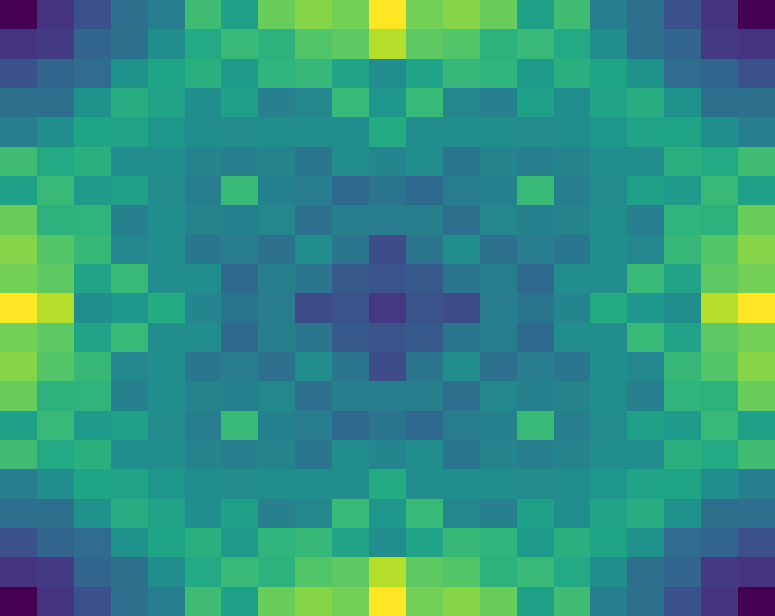}\end{minipage} &
        \begin{minipage}{0.18\textwidth}\centering \textbf{$n=22$}\\[2pt]
        \includegraphics[width=\linewidth]{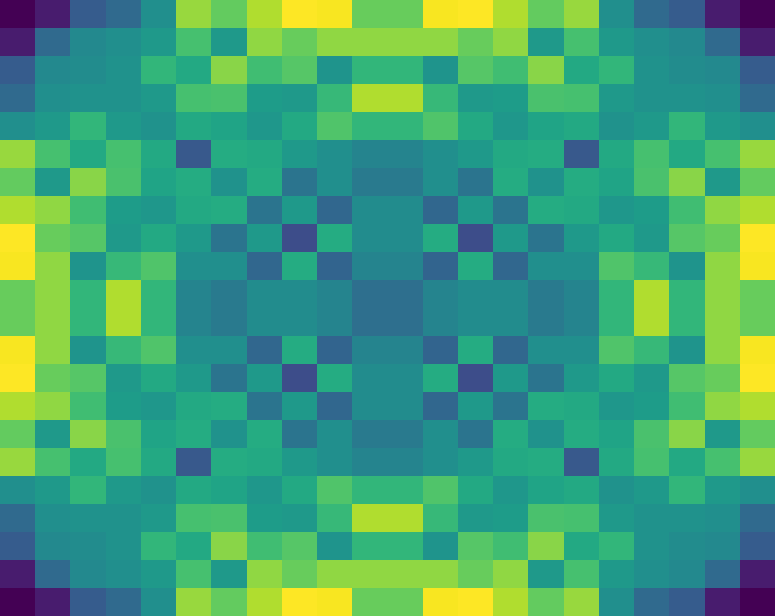}\end{minipage} &        
        \begin{minipage}{0.18\textwidth}\centering \textbf{$n=23$}\\[2pt]
        \includegraphics[width=\linewidth]{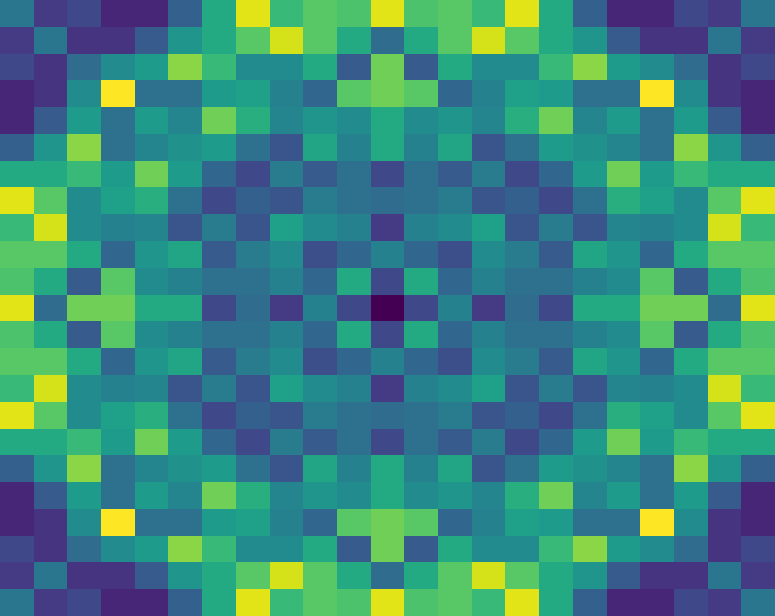}\end{minipage} &
        \begin{minipage}{0.18\textwidth}\centering \textbf{$n=24$}\\[2pt]
        \includegraphics[width=\linewidth]{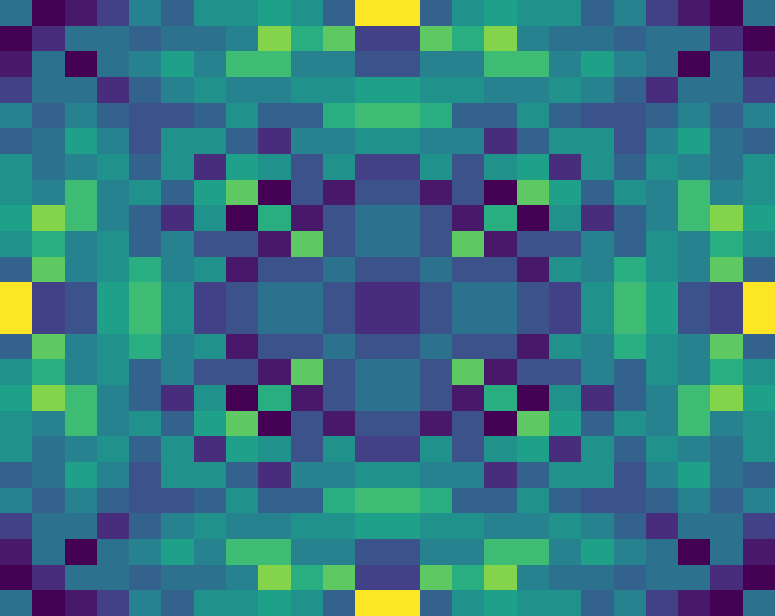}\end{minipage} \\
        \begin{minipage}{0.18\textwidth}\centering \textbf{$n=25$}\\[2pt]
        \includegraphics[width=\linewidth]{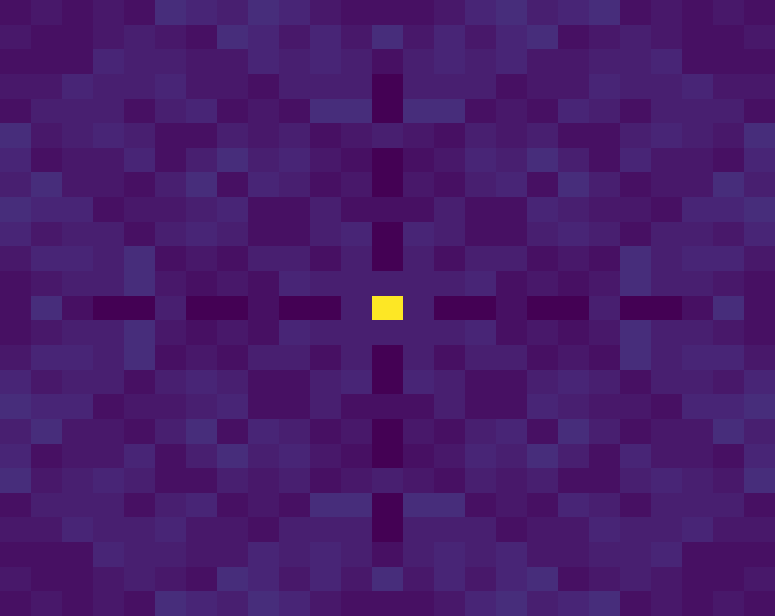}\end{minipage} &
        \begin{minipage}{0.18\textwidth}\centering \textbf{$n=26$}\\[2pt]
        \includegraphics[width=\linewidth]{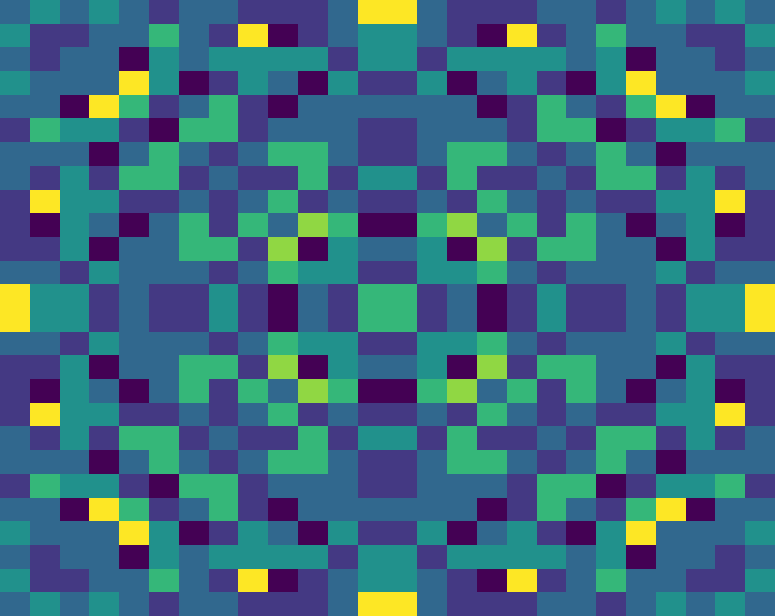}\end{minipage} &
        \begin{minipage}{0.18\textwidth}\centering \textbf{$n=27$}\\[2pt]
        \includegraphics[width=\linewidth]{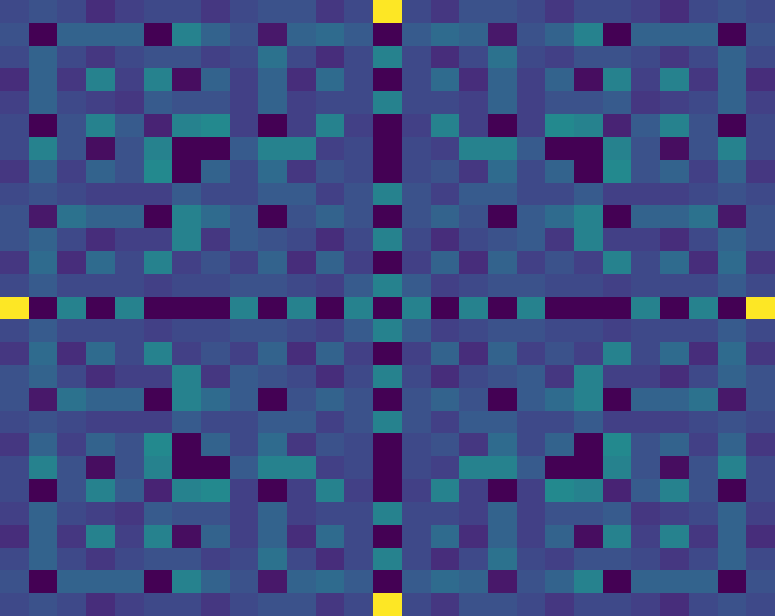}\end{minipage} &        
        \begin{minipage}{0.18\textwidth}\centering \textbf{$n=28$}\\[2pt]
        \includegraphics[width=\linewidth]{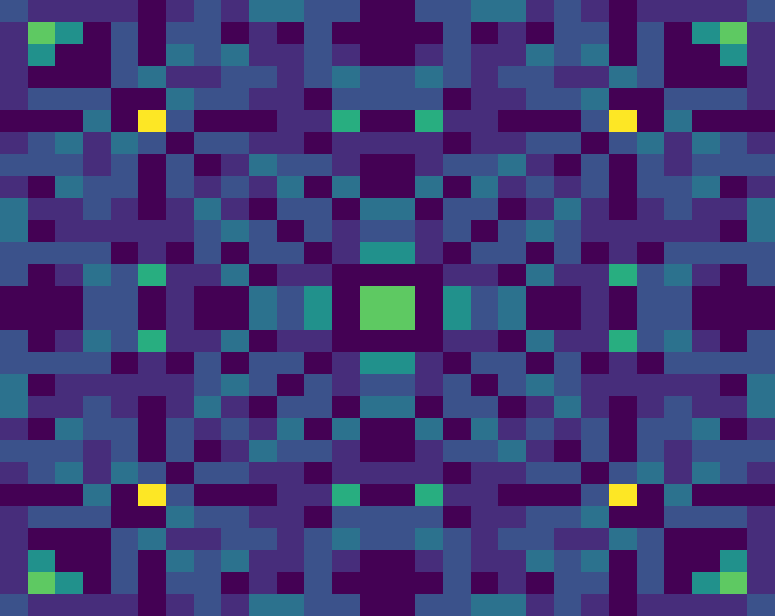}\end{minipage} &
        \begin{minipage}{0.18\textwidth}\centering \textbf{$n=29$}\\[2pt]
        \includegraphics[width=\linewidth]{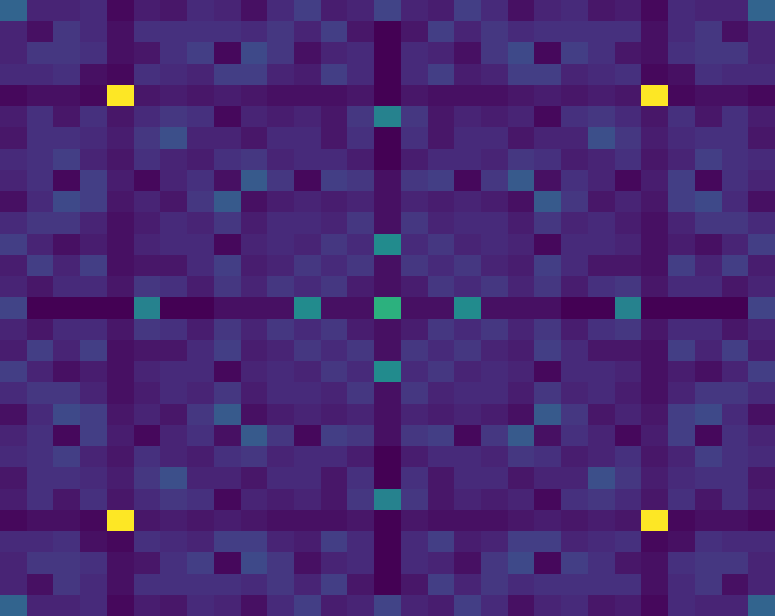}\end{minipage} \\
        \begin{minipage}{0.18\textwidth}\centering \textbf{$n=30$}\\[2pt]
        \includegraphics[width=\linewidth]{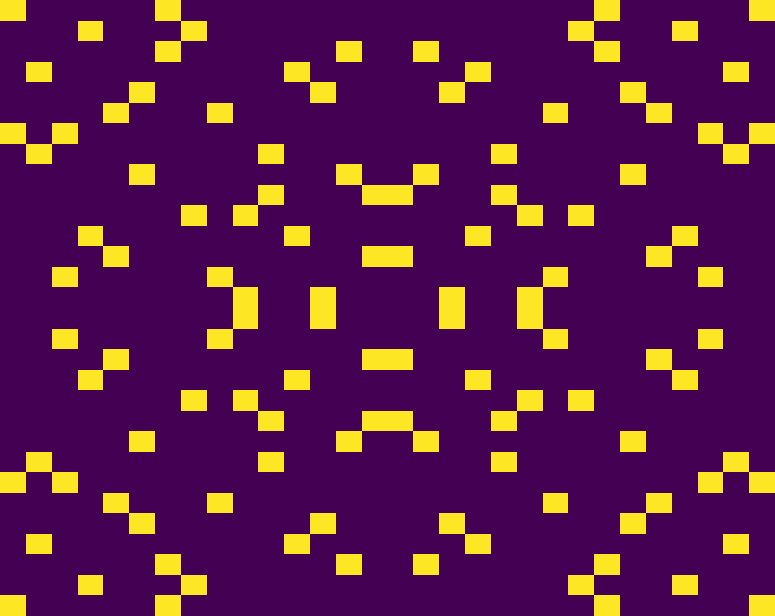}\end{minipage} &
        \begin{minipage}{0.18\textwidth}\centering \textbf{$n=45$}\\[2pt]
        \includegraphics[width=\linewidth]{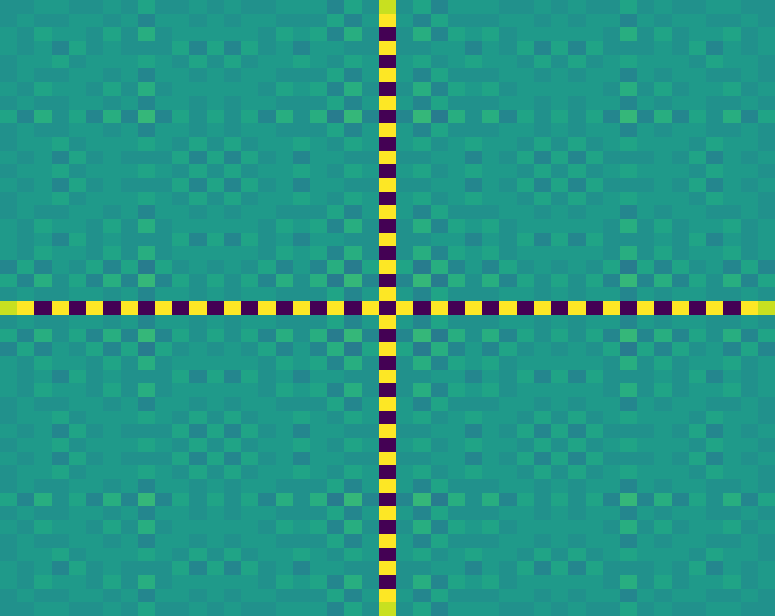}\end{minipage} &
        \begin{minipage}{0.18\textwidth}\centering \textbf{$n=176$}\\[2pt]
        \includegraphics[width=\linewidth]{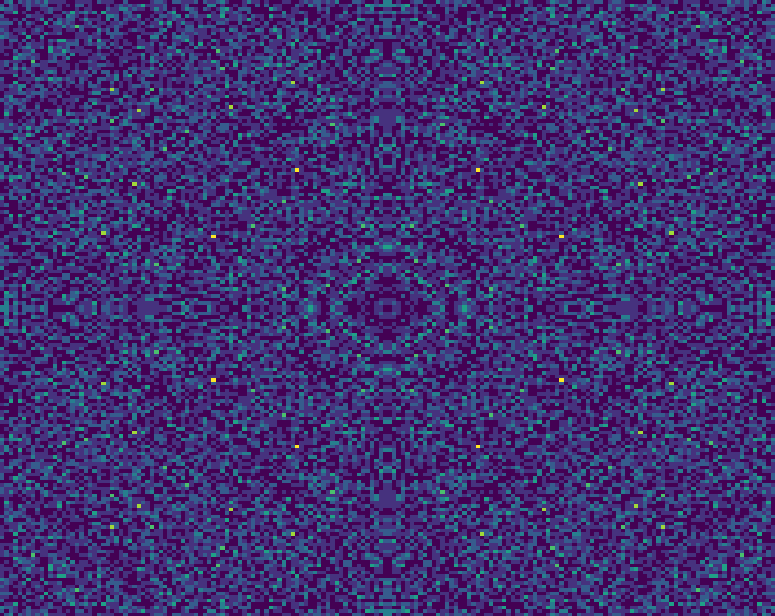}\end{minipage} &     
        \begin{minipage}{0.18\textwidth}\centering \textbf{$n=230$}\\[2pt]
        \includegraphics[width=\linewidth]{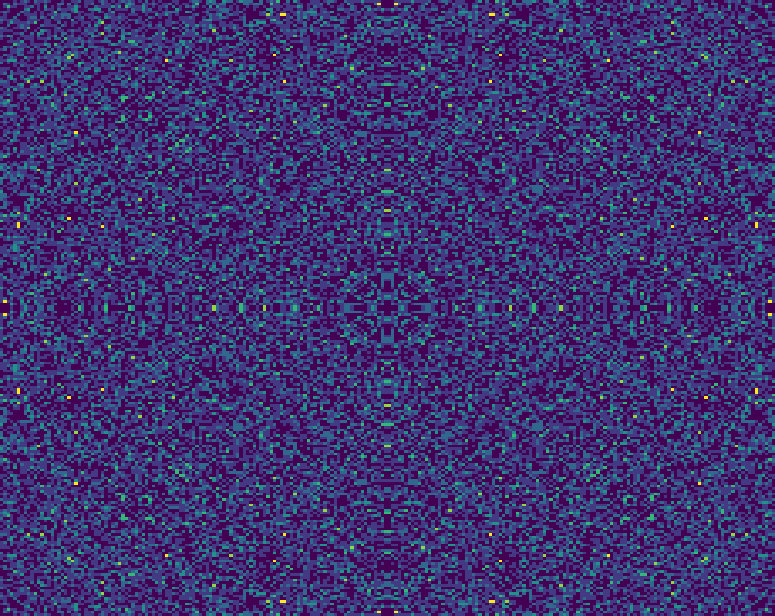}\end{minipage} &
        \begin{minipage}{0.18\textwidth}\centering \textbf{$n=248$}\\[2pt]
        \includegraphics[width=\linewidth]{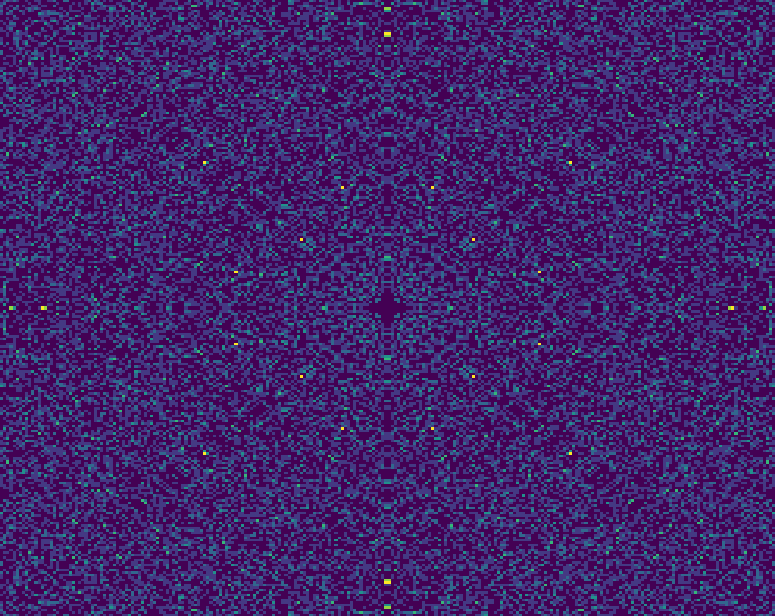}\end{minipage} \\
        \multicolumn{3}{c}{} 
    \end{tabular}
    \caption{Complete UCFMs ($F_n$) for $5 \leq n \leq 29$ (rows 1-5) and incomplete UCFMs for $n={30, 45, 176, 230, 248}$ (only row 6).}
    \label{fig:UCFM_all}
    \vspace{-0.5cm}
\end{figure*}

In Fig. \ref{fig:U/F_Results}, the mean runtimes (left axis, log scale) of the proposed reconstruction method and Russo's method in \cite{russo2010costas} are depicted for $4 \leq n \leq 17$. Each data point of this plot is the average of the $10$ runs for $n < 13$ and $3$ runs for $n \geq 13$ on a Dell PowerEdge T640 server, equipped with dual Intel Xeon Gold 2.1 GHz processors, and 128 GB RAM. The codes were implemented in Python 3.12, and no parallel processing was employed. Fig. \ref{fig:U/F_Results} (right axis) shows the improvement in the mean runtime of the proposed method, calculated as
\(
100 \times \bigl( 1 - E\{t_{U/F}\} / E\{t_{\text{Russo}}\} \bigr),
\)
where \(E\{\cdot\}\) is the expectation operator, and \(t_{U/F}\) and \(t_{\text{Russo}}\) are the runtimes of the proposed and Russo methods, respectively. 

For small orders ($4 \leq n \leq 9$), the performance is significantly higher with respect to Russo's method with a peak improvement of $59\%$ at $n = 5$. After this peak, the performance gain decreases but stabilizes around $15\%$ for $n \geq 13$. Even with only the exploitation of the first column of $F_n$ and the block structure of $U_n$, there is a significant improvement in the performance. While parallel processing can enhance performance, a significant algorithmic improvement can be achieved by fully exploiting the residual frequency structure in $F_n$ to reconstruct $U_n$.

\section{Discussion and Conclusion}
This paper presents a holistic approach using UCMs and UCFMs to analyze and discover new Costas arrays. Their structural properties and theoretical foundations are examined, and a reconstruction-based search method is proposed to obtain UCMs from UCFMs. Numerical results demonstrate significant improvements over the benchmark Russo method, highlighting the efficiency of the proposed approach. Furthermore, Costas arrays obtained through the proposed approach remain suitable for ISAC applications due to their ideal autocorrelation properties for sensing and low cross-correlation properties for multi-user communication.

Our observations and results suggest that any Costas array of order $n$ may exist as a row within a UCM of the same order, rather than independently. This insight indicates that a unified, matrix-based perspective may be more effective than classifying Costas arrays as algebraic or sporadic. Thus, it facilitates a potential column-by-column search method to reconstruct a UCM.

The proposed framework also provides a foundation for AI-driven Costas array discovery. By integrating UCMs/UCFMs with generative AI models, the search for new arrays can be further accelerated. Although data scarcity is a challenge, data augmentation, pre-trained models used for combinational or matrix-completion problems, and generative modeling offer promising strategies for  Costas array discovery.

\bstctlcite{IEEEexample:BSTcontrol}

\bibliographystyle{IEEEtran}
\vspace{-0.15cm}
\bibliography{IEEEabrv,ref_fg_costas}


\end{document}